\documentclass{article}

\PassOptionsToPackage{numbers, compress}{natbib}

\usepackage[preprint]{neurips_2026}

\usepackage[utf8]{inputenc} % allow utf-8 input
\usepackage[T1]{fontenc}    % use 8-bit T1 fonts
\usepackage[colorlinks,citecolor=blue]{hyperref}       % hyperlinks
\usepackage{url}            % simple URL typesetting
\usepackage{booktabs}       % professional-quality tables
\usepackage{amsfonts}       % blackboard math symbols
\usepackage{nicefrac}       % compact symbols for 1/2, etc.
\usepackage{microtype}      % microtypography
\usepackage[dvipsnames]{xcolor}         % colors
\usepackage{amsmath}
\usepackage{amsthm}
\usepackage[linesnumbered,ruled,lined,noend]{algorithm2e}
\usepackage{mathtools}
\usepackage{multirow}
\usepackage{subcaption}
\usepackage{wrapfig}
\usepackage{tabularx}
\usepackage{makecell}

\DeclareMathOperator*{\softmax}{softmax}
\DeclareMathOperator*{\argmax}{arg\,max}

\theoremstyle{plain}
\newtheorem{proposition}{Proposition}
\newtheorem{theorem}{Theorem}

\theoremstyle{remark}

\title{Watermarking Game-Playing Agents in \\ Perfect-Information Extensive-Form Games}

\author{%
  Juho Kim \\
  Computer Science Department \\
  Carnegie Mellon University \\
  \texttt{juhok@cs.cmu.edu} \\
  \And
  Fei Fang \\
  Software and Societal Systems Department \\
  Carnegie Mellon University \\
  \texttt{feifang@cmu.edu} \\
  \And
  Tuomas Sandholm \\
  Computer Science Department, CMU \\
  Strategic Machine, Inc. \\
  Strategy Robot, Inc. \\
  Optimized Markets, Inc. \\
  \texttt{sandholm@cs.cmu.edu} \\
}

\begin{document}

\maketitle

\begin{abstract}
	Watermarking techniques for large language models (LLMs), which encode hidden information in the output so its source can be verified, have gained significant attention in recent days, thanks to their potential capability to detect accidental or deliberate misuse.
	Similar challenges involving model misuse also exist in the context of game-playing, such as when detecting the unauthorized use of AI tools in gaming platforms (\textit{e.g.}, cheating in online chess).
	In this paper, we initiate the study of how game-playing strategies can be watermarked.
    We show how the KGW watermark for LLMs can be adapted to watermark game-playing agents in perfect-information extensive-form games.
	The watermark can then be detected using a statistical test.
	We show that the degradation in the quality of the watermarked strategy profile, quantified by the expected utility, can be bounded, but there is a tradeoff between detectability and quality.
	In our experiments, we bootstrap the watermarking framework to various chess engines and demonstrate that a) the impact of the watermark on the quality of the strategy is negligible and b) the watermark can be detected with just a handful of games.
\end{abstract}

\section{Introduction}
\label{sec:introduction}

In the past decades, breakthroughs in the field of artificial intelligence (AI) led to the development of superhuman AI agents capable of beating top human experts in games like chess~\cite{campbelletal}, Go~\cite{silveretal}, poker~\cite{brownandsandholm2018,brownandsandholm2019}, and dark chess~\cite{zhangandsandholm}.
As strong as they are in performance, game-playing AI agents can potentially grow even stronger as new algorithms develop and computing power increases.
As such, there have been growing concerns about their unauthorized usage, presenting the need for the means to embed and detect a robust, unique signature, \textit{i.e.}, a watermark, in their gameplay.

Watermarking techniques can be useful to communities for games like chess and poker, which face a rampant \textbf{cheating} problem~\cite{chappell,grant}, especially online, as it can be used to detect improper usage of AI.
Additionally, the developers of AI tools themselves can use a watermark for license enforcement and \textbf{intellectual property (IP) protection}.
In 2021, the development team behind the open-source chess engine Stockfish filed a lawsuit against the developers of Houdini 6, a proprietary chess engine, accusing them of code plagiarism and license non-compliance~\cite{stockfish}.
One central piece of evidence included Houdini's leaked source code; however, had a hidden watermark been used beforehand, they would have had another, potentially more irrefutable evidence of IP theft.
A watermark can also deter other types of model theft, such as distillation learning, thanks to its radioactive nature~\cite{sander}: a watermark remains detectable in the outputs of student models trained from a watermarked teacher model.
In another use case, AI researchers interested in mimicking human behavior in games can use a watermark detector for \textbf{data cleaning} to sanitize datasets that are contaminated with AI gameplays.

There exists a vast literature~\cite{dathathrietal,kirchenbaueretal} for watermarking texts generated by large language models (LLMs), that is, encoding hidden signals in LLM outputs so their source can be verified.
In this paper, we propose adapting LLM watermarks for game playing.
Specifically, we explore adapting the KGW watermark, which is the most prominent and widely established method (later in our discussions, we justify this choice and discuss adapting alternative LLM watermarks).
LLM watermarks were developed in response to the meteoric rise of LLMs in popularity and their strong performance in broad natural language processing tasks to address concerns about the potential risks of their accidental or deliberate misuse.
Real-life examples of this include LLM-operated bot networks in social media platforms~\cite{qiaoetal}, large-scale phishing scams~\cite{schulz}, and the spread of misinformation~\cite{lancet}.
Additionally, revelations about the collapse of AI models when trained repeatedly on synthetic data~\cite{shumailovetal} present the necessity to detect LLM-generated outputs in LLM training sets.

Curiously, text generation and game-playing have \textbf{striking similarities}, as summarized in Table~\ref{tab:analogy}, which allow LLM watermarks to be readily adapted to game-playing.
Just as LLMs model the probability distribution over tokens given a context or a prompt as a prior, a strategy profile provides the probability distribution over the set of available actions at a particular history or information set (infoset).
The text generation process ends upon sampling the end-of-sequence (EOS) token.
This is analogous to how game-playing agents apply actions until they reach a terminal history.
\begin{wraptable}{R}{0.5\linewidth}
	\centering
	\caption{Analogies between text generation and game playing.}
	\label{tab:analogy}
	\begin{tabular}{rcl}
		\toprule
		Text generation & & Game playing \\
		\midrule
		LLM & $\longleftrightarrow$ & Strategy profile \\
		Token & $\longleftrightarrow$ & Action \\
		Context/Prompt & $\longleftrightarrow$ & History/Infoset \\
		EOS token & $\longleftrightarrow$ & Terminal history \\
		Text quality & $\longleftrightarrow$ & \makecell{Expected utility \\ /Exploitability} \\
		Sentence entropy & $\longleftrightarrow$ & Mobility/Purity \\
		\bottomrule
	\end{tabular}
\end{wraptable}

In LLM watermarking, a watermark distorts the LLM's modeled probability distribution and thus can potentially degrade the quality of the generated text, typically quantified by calculating the perplexity cost~\cite{kirchenbaueretal} or conducting user feedback studies~\cite{dathathrietal}.
Naturally, such a concern carries over to game playing as well, where a watermark's distortion of action probabilities can reduce an agent's expected utility or make the agent vulnerable to exploitation.

There are certain situations in text generation where watermarking must be forgone, namely when dealing with low sentence entropy.
As an example, consider the Python code snippet `\verb|print("Hello,|', which is most likely followed by `\verb|World!")|'.
Here, the LLM watermark should more or less preserve the LLMs' original probability distributions.
Similarly, a watermarked strategy profile should retain its original strategy in situations with low mobility or when the action probability distribution is pure or almost pure.
A real-life example of this would be in chess openings or endgames, where the set of correct actions is limited and playing incorrect actions can be detrimental.

However, \textbf{notable differences} exist between text generation and game playing.
Whereas in text generation, a single LLM generates the text, in game-playing, multiple players, each employing their own strategies, apply actions sequentially.
Also, while an attacker in text generation can attempt to remove the watermark as a post-processing step, an analogous attack in game playing is impractical because a) undoing an action is usually not allowed in competitive settings, b) other players and, if present, a neutral observer can record the gameplay, and c) utilities are typically realized to players as soon as the game ends.
In short, offline watermark removal attacks, although technically possible, are not relevant in game playing.
Therefore, the attacker is limited to online attacks only, that is, removing the watermark for a single pending action each time the attacker is in turn to act.

That said, game playing also introduces \textbf{new challenges} that are not present in text generation.
While a typical LLM user may be satisfied with simply having the tokens sampled to them, a user of a game solver or an engine may not be content with simply having actions sampled to them.
Indeed, they may also be interested in exploring alternative actions and comparing them with the recommended action using (estimated) expected utilities.
Thus, for our use case, the watermark must be expanded so it can also return the associated expected utilities for different actions or histories.

\subsection{Our contributions}

In this paper, we initiate the study of watermarking strategies for gameplay. We show how the widely established KGW watermark technique~\cite{kirchenbaueretal} for LLMs can be adapted for perfect-information extensive-form games.
We also describe how to address the new challenges introduced by this setting, namely, returning the (estimated) expected utilities without compromising the watermark.
We show that, although watermarking can distort the action probability distribution, the degradation in the quality of the watermarked strategy profile, quantified by the expected utility, can be bounded; however, there exists a tradeoff between detectability and quality.
In our experiments, we bootstrap the watermark to mainstream chess engines to a) demonstrate that the Elo differences between the chess engines and their watermarked counterparts are negligible, b) show that the watermark can be detected with just a handful of games, and c) ablate on the watermark parameters.

\section{Background and related work}

We review game-theoretic concepts, LLM watermarking, and cryptographic applications to games.

\subsection{Perfect-information extensive-form games (EFGs)}

A \textit{perfect-information extensive-form game (EFG)} contains a set of players $P$, which includes the chance player $p_c$, and a set of histories $H$.
A history is a sequence of actions, and each history $h$ is associated with a player in turn to act $p(h)$ and a set of available actions $A(h)$.
We use $h \cdot a$ to denote applying $a \in A(h)$ at $h$.
Also, we use $h \sqsubseteq h'$ to denote that $h$ is a prefix of $h'$.
When $p(h) = p_c$, \textit{i.e.}, $h$ is a chance history, then $f_c(h, \cdot)$ gives a probability distribution over $A(h)$.
$Z \subseteq H$ is a set of terminal histories.
Every terminal history $z \in Z$ is associated with $|P|$ utilities (one for each player), and we use $u_i(z)$ to denote the value for $i \in P$.

Each player $i \in P$ plays with strategy $\sigma_i$ that assigns a probability distribution over the set of available actions in every history $h$ where $p(h) = i$.
A strategy profile $\sigma$ is defined as a $|P|$-tuple of strategies from all players.
We denote by $u_i(\sigma)$ the expected utility of $i$ when players play according to $\sigma$.
We use $\sigma_{-i}$ to denote the strategies of all players except that of $i$.
A best response of $i$ to $\sigma_{-i}$ is denoted as $BR(\sigma_{-i})$ and satisfies $u_i(BR(\sigma_{-i}), \sigma_{-i}) = \max_{\sigma_i'}{u_i(\sigma_i', \sigma_{-i})}$.
A traditional solution concept in perfect-information EFGs is a subgame perfect Nash equilibrium (SPNE) where no player stands to gain by best responding individually, even in states reached with zero probabilities.

An SPNE can be calculated using backward induction or approximated using algorithms such as alpha-beta pruning.
Backward induction is a dynamic programming technique that propagates game utilities from the leaves to the root of the game tree.
At each parent node, the corresponding player selects a child node that yields the maximum expected utility and sets the parent's values as the child's.
In the alpha-beta pruning family of algorithms, with a modern example being AlphaZero~\cite{silveretal}, the expected utility is replaced with its heuristic estimate after a certain depth, and pruning is performed on the game tree to cut down the number of explored nodes.
Typically, the heuristic function is either hand-crafted or is underlied by a neural network trained via reinforcement learning.

\subsection{Watermarking LLMs}
\label{subsec:watermarking-llms}

We begin by formalizing \textit{text generation}.
Let $\mathcal{V}$ be the vocabulary, \textit{i.e.}, the set of tokens (words or parts of words).
An LLM is designed so that, given a sequence of tokens $v_1, v_2, \hdots, v_{t - 1} \in \mathcal{V}$ as a prompt, outputs a probability distribution $\vec{p} \in \Delta^\mathcal{V}$ over $\mathcal{V}$.
From $\vec{p}$, the next token $v_t$ is sampled, which is then appended to the input sequence.
The resulting sequence is passed back to the LLM to obtain $v_{t + 1}$.
The text generation process continues in a loop until an EOS token is encountered.

While there are several methods to watermark LLM outputs~\cite{dathathrietal,kirchenbaueretal,changetal}, we center our focus on the widely established \textit{KGW watermark}.
The KGW watermark is flexible in that it can not only be bootstrapped to an existing LLM as a proxy but also be made public, \textit{i.e.}, all of its parameters can be released to the public, allowing any interested third parties to detect the watermark without any access to the LLM.
Optionally, the watermark can be employed privately with a secret key, so the detection service can be provided through a secure API.
In addition, the KGW watermark can be detected with a statistical test.
By adapting this watermark for game playing, we can inherit its desirable properties.

The pseudocode of the KGW watermark is available in Appendix~\ref{sec:kgw-watermark}.
It accepts the following parameters: a \textcolor{Green}{\textbf{green}} list size $\gamma \in (0, 1)$ and hardness $\delta \ge 0$.
The watermark begins by querying the LLM with the input prompt to obtain logits $\vec{l} \in \mathbb{R}^\mathcal{V}$ over the vocabulary.
Then, the hash of $v_{t - 1}$ is used to seed a pseudo-random number generator.
The generator is then used to randomly partition $\mathcal{V}$ into a \textcolor{Green}{\textbf{green}} list $G$ of size $\gamma|\mathcal{V}|$ and a \textcolor{Red}{\textbf{red}} list $R$ of size $(1 - \gamma)|\mathcal{V}|$.
The watermark then modifies $\vec{l}$ to introduce a bias toward selecting a token in $G$ as opposed to $R$ during the sampling process.
This is achieved by adding $\delta$ to the logits corresponding to the tokens in $G$.
Finally, $\vec{l}$ is passed to a softmax function to obtain probabilities $\vec{p} \in \Delta^\mathcal{V}$, from which $v_t$ is sampled.

\subsection{Steganography in game-playing}

Watermarking is a special case of \textit{steganography}, which encompasses methods of hiding secret data in a public medium such as text and images.
One medium of secret communication explored by several prior works is gameplay data~\cite{changandechizen,hernandezcastroetal,mandujanoetal}, but their techniques are tailored for encoding \textbf{generic} secret data, and thus are not specialized enough for watermarking.
Plus, a notable difference between our contribution and the prior works of~\citet{mandujanoetal} and \citet{changandechizen}~is that their algorithms are not designed to be bootstrapped to existing player policies.
Moreover,~\citet{mandujanoetal} control both the player and the game environment (\textit{i.e.}, the nature player), which is not representative of typical game-theoretic settings.
Additionally, unlike ours, these prior works do not provide any game-theoretic guarantees, such as bounds on the loss in expected utility.

\section{Methodology}

We are now ready to present our results.
We begin by discussing how one can adapt the KGW watermark~\cite{kirchenbaueretal} for strategy profiles in perfect-information EFGs.

\subsection{Watermarking strategy profiles}

\begin{algorithm}[t!]
	\caption{KGW watermark adapted for game playing}
	\label{alg:watermark}
	\textbf{Input:} strategy profile $\sigma$, \textcolor{Green}{\textbf{green}} list size $\gamma \in (0, 1)$, and hardness $\delta \ge 0$.
	\BlankLine
	\textsc{NextAction}(history $h \in H$): \\
	\Indp{
		Query the underlying model of $\sigma$ for the expected utilities $\vec{u} \in \mathbb{R}^{A(h)}$ over available actions at $h$. \\
		\BlankLine
		Seed a pseudo-random number generator with the hash of the observation made at $h$. \\
		\BlankLine
		Randomly partition $A(h)$ into a \textcolor{Green}{\textbf{green}} list $G$ of size $\gamma|A(h)|$ and a \textcolor{Red}{\textbf{red}} list $R$ of size $(1 - \gamma)|A(h)|$. \\
		\BlankLine
		$\vec{v} \coloneq \vec{u} + \left(\begin{cases} \delta & a \in G \\ -\gamma\delta/(1 - \gamma) & a \in R \end{cases}\right)_{a \in A(h)}$\;
		$a^* \coloneq \argmax_{a \in A(h)}{\vec{v}_a}$\;
		\Return $a^*$\;
	}
	\Indm
\end{algorithm}

Whereas, in text generation, the watermark acts as a proxy to an LLM, in game playing, it acts as a wrapper to a strategy profile $\sigma$, calculated by an underlying game-solving algorithm.
Akin to querying the LLM for logits over the vocabulary, our watermark queries the underlying model of $\sigma$ with history $h$ to calculate (or estimate) the expected utilities over the set of available actions $A(h)$.

We then use the hash of the observation made at $h$ to seed the pseudo-random number generator.
This deviates from the original KGW algorithm, which used the final token of the prompt (an analogue of the final action) to seed the pseudo-random number generator.
We begin by justifying why we \textbf{can} make such a change.
In the original KGW watermark, the motivation for seeding with the final token was to make the watermark robust against splicing attacks and hence remain detectable even with a contiguous substring of the generated text.
In game-playing, this consideration is no longer relevant, as splicing attacks are impossible due to the reasons stated in the introduction (\textit{e.g.}, presence of a neutral/adversarial observer, impossibility of offline attacks, \textit{etc.}).

We continue by explaining why we \textbf{need} to make such a change.
Applying the approach of \citet{kirchenbaueretal}~directly in game playing can be dangerous, as the attacker, trying to get a recommendation for $h$, can pass another history $h'$ with a different final action that nonetheless produces the same observation.
By doing so, the attacker would be able to get a recommendation for $h$ while bypassing the watermark.
That said, using the observation at $h$ instead of the history itself is not fully robust either, as the attacker can instead pass another history $h''$ that results in a similar albeit different observation.
While the attacker can theoretically bypass the watermark this way, later in the paper, we discuss why this is risky for the attacker and propose how such an attack can be countered.

Then, we use the pseudo-random number generator to randomly partition the available actions into a \textcolor{Green}{\textbf{green}} list $G$ of size $\gamma|A(h)|$ and a \textcolor{Red}{\textbf{red}} list $R$ of size $(1 - \gamma)|A(h)|$, and add $\delta$ to the (estimated) expected utilities corresponding to the actions in $G$ while subtracting $\gamma\delta/(1 - \gamma)$ to those in $R$.
Here, unlike the KGW watermark, we also subtract from the values corresponding to the \textcolor{Red}{\textbf{red}}-list actions to ensure that the adjustment to every (estimated) expected utility is zero in expectation.
Finally, the action yielding the maximum adjusted (estimated) expected utility is selected, which is then returned.
The pseudocode of the watermarking process is shown in Algorithm~\ref{alg:watermark}.
When $\delta \to \infty$, the watermark only returns \textcolor{Green}{\textbf{green}}-list actions, and, when $\delta \to 0^+$, the watermark only plays a tiebreaking role.

\subsection{Detecting the watermark}

As in the KGW watermark, the null hypothesis is that a particular player applies actions with no knowledge of the \textcolor{Green}{\textbf{green}} and \textcolor{Red}{\textbf{red}} lists.
Let $n$ be the total number of actions from the player and $n_G$ be the number of \textcolor{Green}{\textbf{green}}-list actions from the player.
We have the following z-score for an arbitrary $\gamma$:
\[
	z = \frac{n_G - \gamma n}{\sqrt{n \gamma (1 - \gamma)}}.
\]
The z-score reflects how many more (or fewer) \textcolor{Green}{\textbf{green}}-list actions are applied than the number expected from a player oblivious to the watermark.
\citet{kirchenbaueretal}~suggested the z-score threshold of $z = 4$, corresponding to approximately $3 \times 10^{-5}$ probability of getting a false positive.

\subsection{Watermarking expected utilities}

\begin{algorithm}[t!]
	\caption{Watermark for expected utility}
	\label{alg:expected-utility}
	\textbf{Input:} strategy profile $\sigma$, \textcolor{Green}{\textbf{green}} list size $\gamma \in (0, 1)$, and hardness $\delta \ge 0$.
	\BlankLine
	\textsc{ExpectedUtility}(history $h \in H$): \\
	\Indp{
		Query the underlying model of $\sigma$ for the expected utility $u$ at $h$. \\
		\BlankLine
		$v \coloneq u$\;
		\BlankLine
		\For{each $h' \cdot a \sqsubseteq h$}{
			\If{\textsc{NextAction} $h'$) partitions $A(h')$ s.t. $a \in G$}{
				$v \coloneq v + \delta$\;
			}
			\Else{
				$v \coloneq v - \gamma\delta/(1 - \gamma)$\;
			}
		}
		\BlankLine
		\Return $v$\;
	}
	\Indm
\end{algorithm}

The main watermark introduced above can be thought of as a black box that takes a history as an input and outputs a \textbf{recommended action}.
As alluded to in the introduction, in game playing, a user may not be satisfied with just being provided with the recommended action.
Indeed, they may also want the associated expected utility for the recommendation.
Moreover, they may want to compare the expected utility of applying this action with that of applying an alternative action as well.
This presents a need for another watermark that takes a history as an input and outputs its \textbf{expected value}.

One cannot simply return the (estimated) expected utility from the underlying model of $\sigma$ without any modification, as an attacker would be able to bypass the watermark by picking an action with the maximum expectation.
A sensible option is to modify the expected utility with the hardness parameter $\delta$ as shown in Algorithm~\ref{alg:expected-utility}.
This is consistent with our main watermark (shown in Algorithm~\ref{alg:watermark}).

However, when this watermarking algorithm is public and $\delta$ is known to the attacker, the attacker would be able to undo the watermark adjustments to obtain the true expected utility returned by the underlying model.
As a consequence, if information about the expected utility is provided to users, the watermark must be deployed privately.
This caveat, although not explored by \citet{kirchenbaueretal}, is also highly relevant in text generation: the KGW watermark must be deployed privately when, alongside sampled tokens, the logits or the probability distributions over the vocabulary are also returned.
The expanded version of the conference paper by~\citet{kirchenbaueretal} discusses how to deploy the KGW watermark privately, and the same technique can be applied to game playing.

\section{Analysis}

In this section, we analyze the KGW watermark adapted to game playing.
For the purpose of analysis, we assume that a) the underlying model of the strategy profile being watermarked gives the true expected utility, b) whether any chosen action is in the \textcolor{Green}{\textbf{green}} or \textcolor{Red}{\textbf{red}} list is independent of the types of other chosen actions, and c) the probability of a watermarked agent selecting a \textcolor{Green}{\textbf{green}}-list action is constant.
Then, we have the following theorem:

\begin{theorem}
	Given a strategy profile $\sigma$, let $\sigma^*$ be its watermarked counterpart, following Algorithm~\ref{alg:watermark}, and suppose that an agent follows $\sigma^*$.
	Furthermore, suppose that whether the watermarked agent chooses a \textcolor{Green}{\textbf{green}}- or \textcolor{Red}{\textbf{red}}-list action at any decision point is independent of the choices made at other decision points, and that the probability of the agent selecting a \textcolor{Green}{\textbf{green}}-list action is a constant $p$.
	Let $n$ be the number of actions played by the agent, and let $L$ be the total loss in utility the agent incurs by applying the watermark.
	Then, for any $t > 0$,
	\[
		\Pr[L < \delta (1 + \gamma/(1 - \gamma)) (pn + t)] > 1 - e^{-2t^2/n^2}.
	\]
	\label{thm:concentration-bound}
\end{theorem}
The proof is in Appendix~\ref{sec:concentration-bound}. Relaxing the independence assumption and the assumption that the probability of selecting a \textcolor{Green}{\textbf{green}}-list action is constant, we have the following proposition for the deterministic bound in the loss in utility an agent incurs by applying the watermark:

\begin{proposition}
	Given a strategy profile $\sigma$, let $\sigma^*$ be its watermarked counterpart, following Algorithm~\ref{alg:watermark}, and suppose that an agent follows $\sigma^*$.
	Let $n$ be the number of actions played by the agent, and let $L$ be the total loss in utility the agent incurs by applying the watermark.
	Then,
	\[
		L \le n \delta (1 + \gamma/(1 - \gamma)).
	\]
	\label{pro:bound}
\end{proposition}
The proof is in Appendix~\ref{sec:bound}.
In general, there exists a \textbf{tradeoff} between detectability and quality in the watermarked strategy profile.
While increasing $\delta$ can help detect the watermark more quickly, \textit{i.e.}, with fewer actions played, doing so causes the agent to potentially lose out on the expected utility.

While the potential loss in expected utility incurred by applying the watermark may be concerning, it may not be as relevant when the underlying model of $\sigma$ only \textbf{estimates} the expected utility (which applies to solvers for large perfect-information EFGs), since the loss in the true expected utility from the inaccuracy of the model can be much greater than the loss incurred by applying the watermark.
In this regard, the watermark can be viewed as simply adding tiny additional noise to the already noisy estimates.
Later in our experiments, we empirically demonstrate that the performance degradation in major chess engines due to watermarking is insignificant.

\section{Experiments}

\begin{table}[t!]
	\small
	\centering
	\caption{
		Results from the UCI engine matches ($\gamma = 0.25$, $\delta = 0.5$ centipawn).
		A negative Elo difference denotes that the watermarked engine is weaker and vice versa.
		LOI is the likelihood that the watermarked engine is weaker than the original.
	}
	\label{tab:experiment}
	\begin{tabular}{l|cccccc}
		\toprule
		\multirow{2}{*}{Engine} & \multirow{2}{*}{Elo difference} & \multirow{2}{*}{LOI} & \multirow{2}{*}{\# draws} & \multicolumn{2}{c}{z-score} & \multirow{2}{*}{AUC} \\
		& & & & NW & W & \\
		\midrule
		asmFish 9 & $-$10.4$\pm$55.3 & 64.5\% & 35 & $-$2.14 & 22.9 & 0.793 \\
		Dragon & $+$10.4$\pm$47.9 & 33.4\% & 51 & 0.151 & 47.0 & 0.790 \\
		Ethereal 14.40 & $-$10.4$\pm$52.6 & 65.2\% & 41 & $-$0.865 & 33.4 & 0.783 \\
		PlentyChess 7.0.0 & $-$6.9$\pm$33.4 & 65.8\% & 76 & 0.262 & 19.4 & 0.715 \\
		Stash v37.0 & $-$20.9$\pm$44.3 & 82.3\% & 58 & $-$0.0464 & 28.8 & 0.787 \\
		Stockfish 17.1 & $+$24.4$\pm$31.0 & 6.3\% & 79 & 0.894 & 28.8 & 0.766 \\
		\bottomrule
	\end{tabular}
\end{table}

\begin{figure}[t!]
	\centering
	\hfill
	\begin{subfigure}{.258\textwidth}
		\centering
		\includegraphics[width=\linewidth]{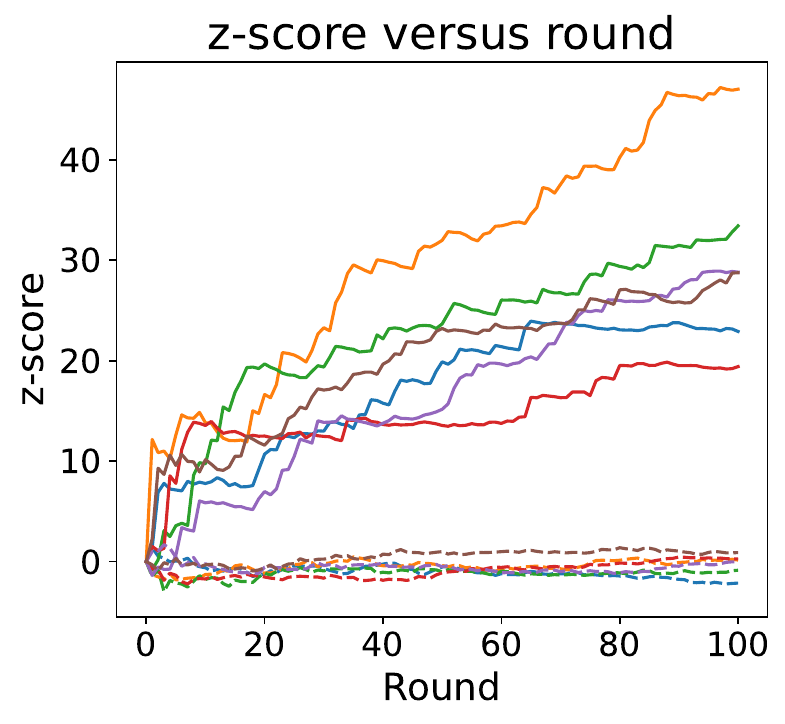}
	\end{subfigure}
	\begin{subfigure}{.284\textwidth}
		\centering
		\includegraphics[width=\linewidth]{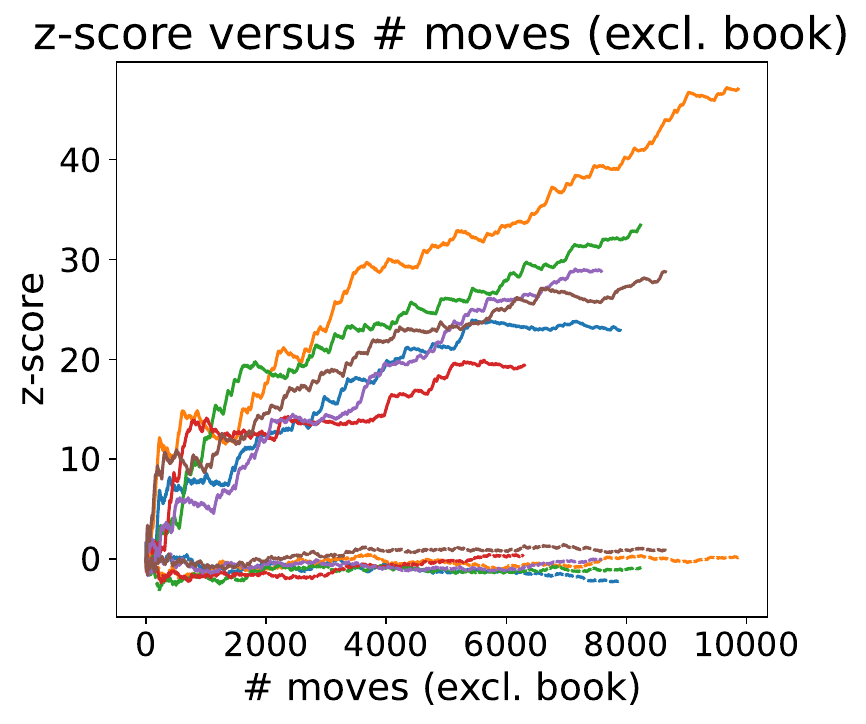}
	\end{subfigure}
	\begin{subfigure}{.411\textwidth}
		\includegraphics[width=0.99\linewidth]{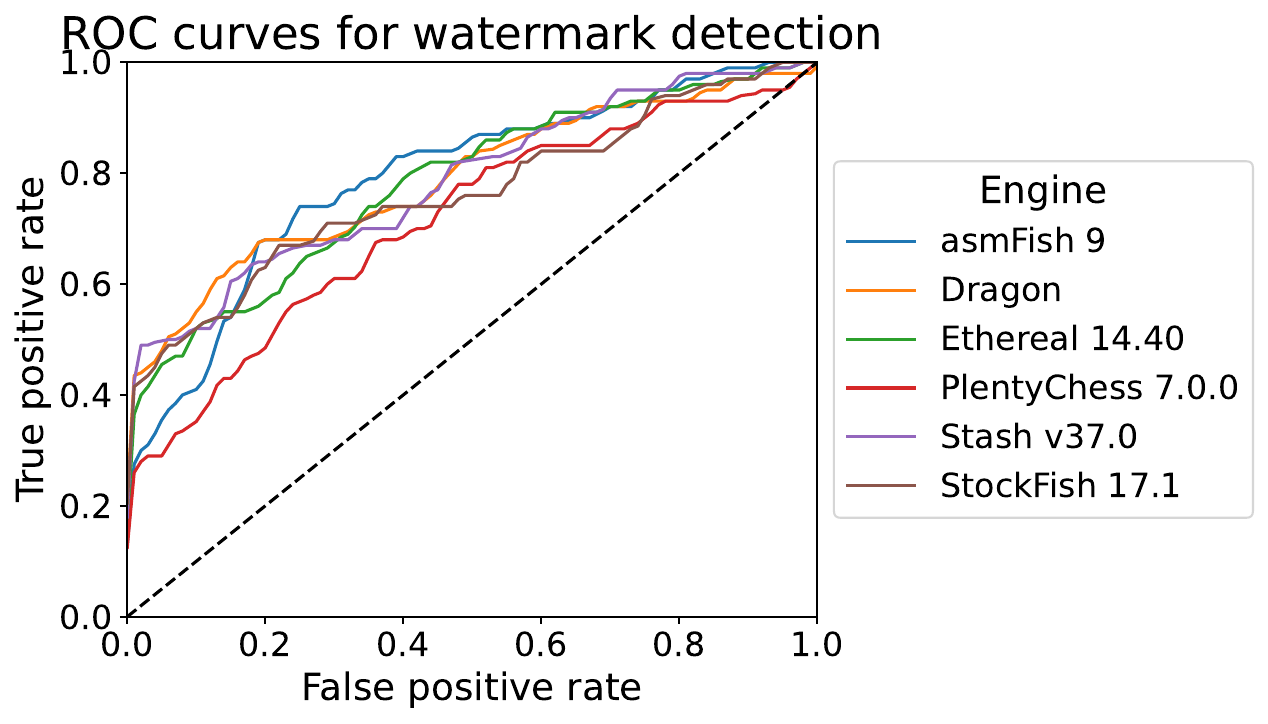}
	\end{subfigure}
	\hfill
	\caption{
		On the left and in the middle, respectively, the z-scores over the number of rounds and moves, respectively.
		On the right, the ROC curves for watermark detection (for each round).
		Solid and dotted lines, respectively, represent watermarked and original chess engines, respectively.
		The AUC values of the ROC curves are tabulated in Table~\ref{tab:experiment}.
	}
	\label{fig:experiment}
	\label{fig:experiment-auc}
\end{figure}

\begin{table}[t!]
	\small
	\centering
	\caption{
		Results from the Stockfish chess engine matches under parameter ablation.
		A negative Elo difference denotes that the watermarked engine is weaker and vice versa.
		LOI is the likelihood that the watermarked engine is weaker than the original.
		The asterisked values (with $^*$) are from the same match, and are also tabulated in Table~\ref{tab:experiment} as entries for Stockfish.
	}
	\label{tab:ablation}
	\begin{tabular}{cc|cccccc}
		\toprule
		\multirow{2}{*}{$\gamma$} & \multirow{2}{*}{$\delta$} & \multirow{2}{*}{Elo difference} & \multirow{2}{*}{LOI} & \multirow{2}{*}{\# draws} & \multicolumn{2}{c}{z-score} & \multirow{2}{*}{AUC} \\
		& & & & & NW & W & \\
		\midrule
		0.1 & 0.5 & $+$17.4$\pm$36.7 & 17.7\% & 71 & $-$0.665 & 7.87 & 0.671 \\
		0.25 & 0.5 & $+$24.4$\pm$31.0$^*$ & 6.3\%$^*$ & 79$^*$ & 0.894$^*$ & 28.8$^*$ & 0.766$^*$ \\
		0.5 & 0.5 & $-$20.9$\pm$30.3 & 91.0\% & 80 & $-$0.693 & 14.7 & 0.728 \\
		0.75 & 0.5 & $-$3.5$\pm$34.1 & 57.9\% & 75 & 0.539 & 6.60 & 0.656 \\
		0.9 & 0.5 & $+$6.9$\pm$37.4 & 35.8\% & 70 & $-$1.41 & 2.10 & 0.592 \\
		\midrule
		0.25 & 0.5 & $+$24.4$\pm$31.0$^*$ & 6.3\%$^*$ & 79$^*$ & 0.894$^*$ & 28.8$^*$ & 0.766$^*$ \\
		0.25 & 1 & $+$10.4$\pm$32.7 & 26.6\% & 77 & 0.950 & 26.2 & 0.760 \\
		0.25 & 2 & $+$17.4$\pm$29.6 & 12.6\% & 81 & 0.224 & 34.3 & 0.850 \\
		0.25 & 5 & $-$24.4$\pm$39.1 & 88.8\% & 67 & $-$0.376 & 42.4 & 0.917 \\
		0.25 & 10 & $-$88.7$\pm$41.7 & 100.0\% & 61 & 0.600 & 40.5 & 0.962 \\
		\bottomrule
	\end{tabular}
\end{table}

\begin{figure}[t!]
	\centering
	\begin{minipage}[t]{.485\textwidth}
		\centering
		\hfill
		\begin{subfigure}{.485\textwidth}
			\centering
			\includegraphics[width=\linewidth]{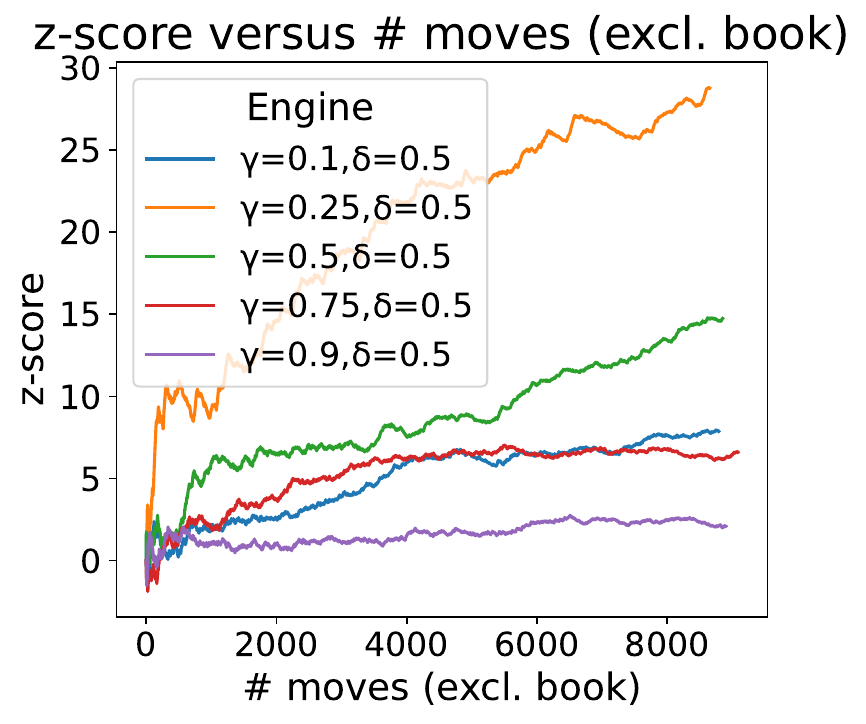}
		\end{subfigure}
		\begin{subfigure}{.485\textwidth}
			\centering
			\includegraphics[width=\linewidth]{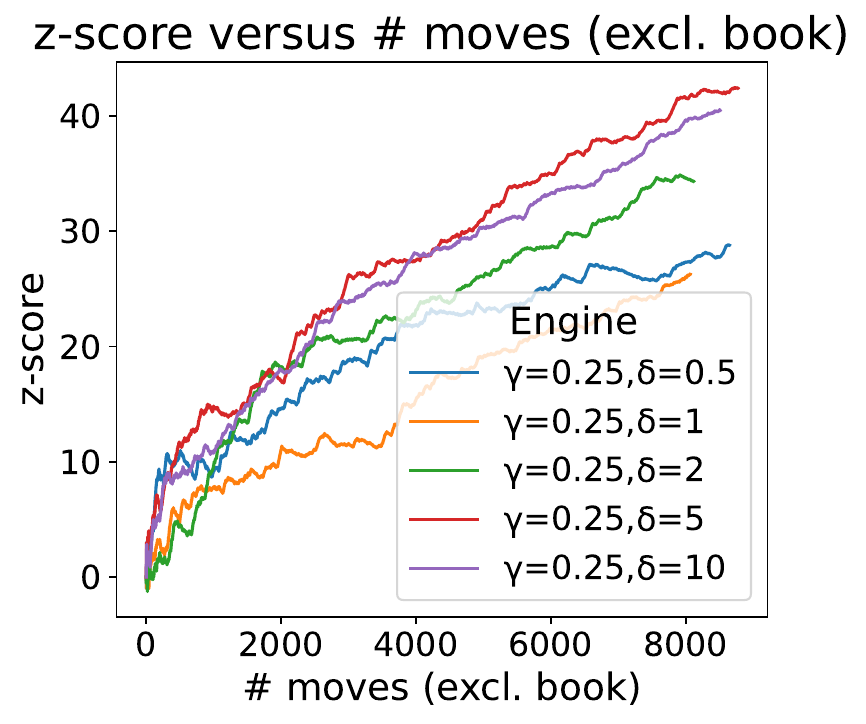}
		\end{subfigure}
		\captionsetup{margin={1em, 1em}}
		\caption{
			z-scores for watermarked Stockfish 17.1.
			The left and right plots ablate on $\gamma$ and $\delta$, respectively.
		}
		\label{fig:ablation}
	\end{minipage}
	\hfill
	\begin{minipage}[t]{.485\textwidth}
		\centering
		\begin{subfigure}{.485\textwidth}
			\centering
			\includegraphics[width=\linewidth]{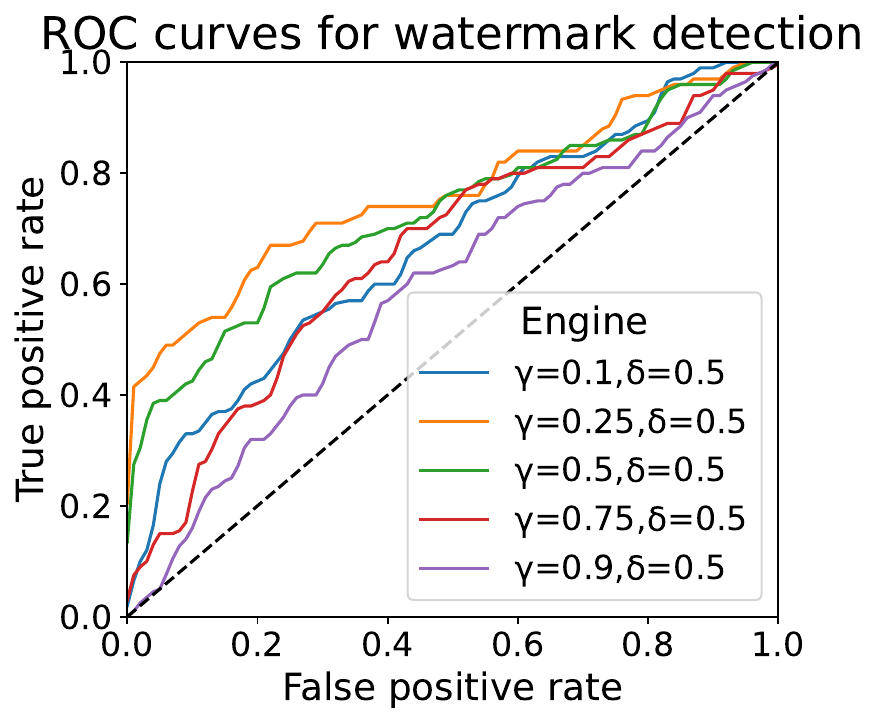}
		\end{subfigure}
		\begin{subfigure}{.485\textwidth}
			\centering
			\includegraphics[width=\linewidth]{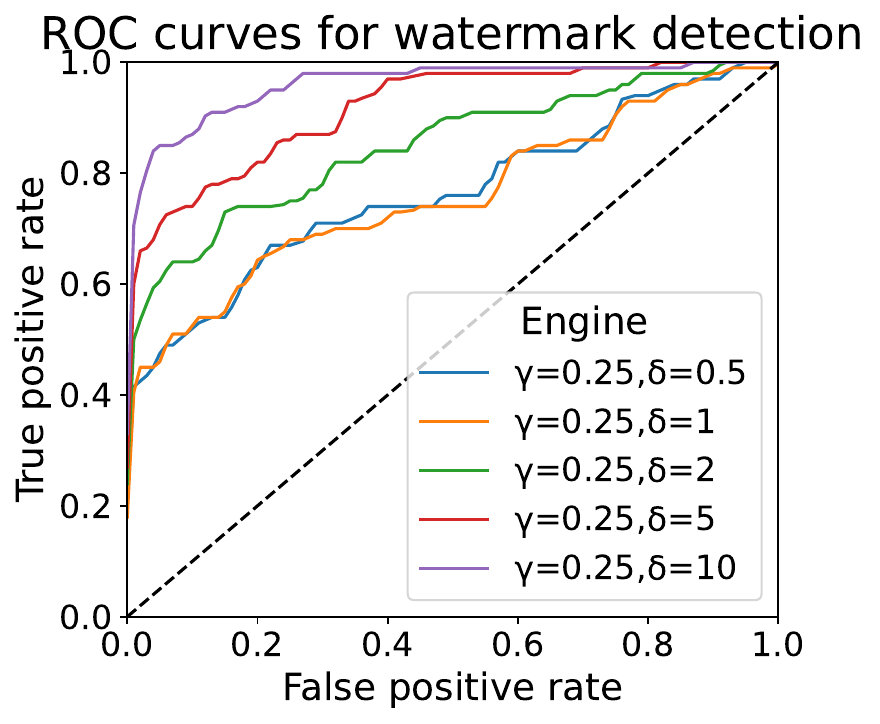}
		\end{subfigure}
		\hfill
		\captionsetup{margin={1em, 1em}}
		\caption{
			The ROC curves for watermark detection (for each round).
			The AUC values are tabulated in Table~\ref{tab:ablation}.
		}
		\label{fig:ablation-auc}
	\end{minipage}
\end{figure}

One notable property of the KGW watermark is that it can easily be bootstrapped to existing LLMs.
The same applies to the watermark when adapted from text generation to game playing.
In our experiments, we bootstrapped the watermark to six selections of popular universal chess interface (UCI) engines: asmFish 9, Dragon, Ethereal 14.40, PlentyChess, Stash v37.0, and Stockfish 17.1.
To calculate the degradation in the quality of the strategy profile introduced by watermarking, as quantified by the loss in expected utility, we pitted each chess engine against its watermarked counterpart.
We seek to demonstrate that a) the watermark does not significantly harm the chess engine performance and b) the watermark can be detected after relatively few rounds of play.

In our experiments, each match consisted of 100 rounds (\textit{i.e.}, complete games), beginning from 8-move positions in the Stockfish chess books dataset, under the time control of 40 moves in 60 seconds.
We calculated the Elo difference resulting from the watermark, as well as the likelihood that the watermarked engines are inferior in strength to their original counterparts.
We also calculated draw probabilities and z-scores from the gameplay data.
Making UCI engines output multiple actions typically incurs some computational cost.
To keep a level playing field, we used the same engine parameters for both watermarked and non-watermarked engines.
The watermark parameters were set as $\gamma = 0.25$ and $\delta = 0.5$ centipawn, which we later found to be ideal in our ablation studies.

The experimental results are tabulated in Table~\ref{tab:experiment}.
Even after 100 rounds, the Elo differences between the watermarked and non-watermarked chess engines are well within the margins of error.
Owing to the stochastic nature of computer chess matches, several watermarked engines are seen even to outperform their non-watermarked counterparts.
Additionally, for every chess engine, the likelihood that its watermarked counterpart is inferior (LOI) fails to reach the 90\% confidence level.
We interpret these results as empirical evidence that the watermark does not significantly harm performance.

The z-scores of the watermarked engines, calculated from the entire match, are far above the suggested $z = 4$ threshold, whereas those of the non-watermarked engines do not meet the threshold.
Recall that $z=4$ corresponds to the false-positive rate of approximately $3 \times 10^{-5}$.
Figure~\ref{fig:experiment} shows how the z-scores change over the number of rounds played.
It can be seen that the z-scores of the watermarked chess engines cross the suggested $z = 4$ threshold relatively early into the match.
The Dragon chess engine passes the threshold only after playing a single round, while asmFish 9 and Stockfish 17.1 do so after finishing Round 2.
The rest hit the threshold after Rounds 4, 8, and 9.

Having watermarked chess engines reach the z-threshold in just a few rounds can be advantageous to a game-playing platform for the sake of security and ecology, as it limits the potential harm inflicted on other legitimate players on the platform.
That said, the data suggests that simply monitoring a potentially malicious actor for a single round of chess may not be enough; instead, multiple rounds must be played to ascertain whether or not a player is cheating with statistical significance.
Indeed, for each round played, we swept the z-score thresholds and found that the area under the receiver operating characteristic (ROC) curves (AUCs) averaged between $0.7$ and $0.8$, which is not ideal.
The ROC curves and their AUC values are in Figure~\ref{fig:experiment-auc} and Table~\ref{tab:experiment}.
Later in our ablation studies, we show that the AUC values can be increased, but this penalizes the strength of the watermarked engine.

\section{Ablation Studies}

In our ablation studies, we exclusively use the Stockfish 17.1 chess engine.
First, fixing the hardness to $\delta=0.5$, we iterated over different values of the \textcolor{Green}{\textbf{green}}-list size $\gamma$.
For each set of parameters, we initialized the watermarked chess engine and recorded its match of 100 rounds against the original Stockfish engine.
In our second ablation study, we fixed $\gamma=0.25$, and iterated over different values of $\delta$.
Again, we recorded the statistics of the matches played.

The results are tabulated in Table~\ref{tab:ablation}.
When iterating over the \textcolor{Green}{\textbf{green}}-list size $\gamma$ (with fixed hardness $\delta=0.5$), it can be seen that $\gamma=0.25$ results in the maximum z-score of the watermarked chess engine.
It also results in the maximum AUC of the ROC curve.
Our results show that setting $\gamma=0.25$ is optimal.
More monotonicity is observed when iterating over the hardness $\delta$, with fixed \textcolor{Green}{\textbf{green}}-list size $\gamma=0.25$.
As $\delta$ increases, the LOI approaches $1$, meaning that the watermarked engine is significantly weaker in strength compared to the original Stockfish engine.
An increase in $\delta$ also corresponds to an increase in the z-score of the watermarked engine and the AUC of the ROC curve.
Clearly, increasing $\delta$ facilitates the watermark detection, but it also diminishes the engine's strength.

\section{Discussion}

Here, we justify our choice of adapting the KGW watermark and discuss possible attacks.

\subsection{Use of alternative LLM watermarks}

The primary reason why we chose the KGW watermark for our study was that its logit boosting mechanism makes the watermark particularly applicable to game-playing settings, as boosting the analogue of logits for playing perfect-information EFGs -- expected utilities -- is both simple and intuitive, and, as we have shown in this paper, leads to a highly performant watermark for game-playing agents. Similarly, we would expect the variants~\cite{zhao,takezawaetal,huoetal,wangetal} of the KGW watermark that rely on the core idea of partitioning tokens into \textcolor{Red}{\textbf{red}}- and \textcolor{Green}{\textbf{green}}-lists to remain highly applicable to the game-playing setting, as long as the newly introduced ideas have obvious analogues in the game-playing setting. For example, LTW introduced by~\citet{wangetal} uses sentence embeddings, whose analogues in game playing are history embeddings. That said, there are also LLM watermarks such as SynthID~\cite{dathathrietal}, which operate on tokens sampled from existing LLMs. Since solutions to perfect-information EFGs give pure probability distributions over available actions, the watermark would sample the same action repeatedly; therefore, SynthID and other similar watermarks will fail when adapted to game playing, although a smoothing method can possibly be used to overcome this.

\subsection{Attacking the watermark}
\label{sec:attack}

One advantage of adapting the KGW watermark for use in game playing is that some of its robustness guarantees can be inherited.
This is compounded by the fact that, as we argued in the introduction, game playing prohibits the attacker from modifying previously played actions.
Thus, common offline watermark attack patterns such as the family of paraphrasing, insertion, and replacement attacks are not relevant in this setting.
With that said, we tailored the KGW watermark for the purpose of playing perfect-information games by making a number of modifications to the original algorithm.
In this section, we argue that our changes do not facilitate an attacker to carry out online attacks, which involve the attacker manipulating the watermark each time the attacker is in turn to act.

A notable difference between the KGW watermark and the modified watermark is that we use the observation made at the queried history instead of the last action (or token) to seed the pseudo-random number generator.
To manipulate this, the attacker, prior to querying the watermarked strategy profile, could try to modify the history so that the resulting observation is slightly different from the original observation.
This is risky for the attacker, as, in games like chess, even small differences in piece positions that seem benign at first glance can have large strategic implications at play.
As a result, this attack can severely degrade the quality of the recommendations, which goes against the interest of the attacker.
Also, to counter such an attack, locality-sensitive hashing can be used so that similar observations result in identical hash values.
Since the hash values are identical, the set of available actions will be partitioned identically across similar observations, thus nullifying the attack.

Henceforth, we make the standard assumption~\cite{kirchenbaueretal} that the watermarked strategy profile is deployed privately.
Instead of taking the recommendation from the watermarked strategy profile, the attacker could try to obtain the watermarked (estimated) expected utility for each available action and use it during the action selection.
One possible strategy is selecting the top-k available actions or those above a certain threshold, from which an action is further subsampled to be played.
This strategy will fail, as the attacker's process of sampling the top actions is still biased toward sampling actions in the \textcolor{Green}{\textbf{green}} list.
Among the sampled actions, the attacker has no way of discriminating between \textcolor{Red}{\textbf{red}}- and \textcolor{Green}{\textbf{green}}-list actions, so the watermark will remain detectable in the attacker's gameplay no matter what subsampling strategy is used.
In fact, any strategy that is biased toward selecting actions with higher watermarked (estimated) expected utilities would remain detectable.

To circumvent the watermark, the attacker could avoid using the watermarked strategy profile except in a few crucial situations.
However, as long as the attacker uses the watermarked strategy profile, the gameplay will show that \textcolor{Green}{\textbf{green}}-list actions are played with probability higher than $\gamma$.
As a result, the z-score calculated from the attacker's gameplay will eventually cross the designated z-threshold.
Thus, the watermark will remain detectable, although its detection will require more gameplay data.

\section{Conclusions and future research}

In this paper, we initiated the study of how strategies can be watermarked.
We showed how the widely established KGW watermark~\cite{kirchenbaueretal} for LLMs can be adapted to watermark strategy profiles in perfect-information games.
Game playing introduces new challenges not present in text generation, namely in having the (estimated) expected utilities provided to the user, along with the recommendation.
We showed how the information about (estimated) expected utilities can be watermarked as well.
Also, in our analysis, we bounded the loss in the expected utility introduced by the watermark.
In our experiments, we bootstrapped the watermark onto six prominent chess engines and showed that no significant degradation in the chess engine's strength can be observed.
Then, we discussed possible attacks on the watermark and proposed countermeasures.
Our method can also be applied to single-agent decision making, as it is a special case of perfect-information EFGs.
Possible future work includes extending our work to imperfect-information EFGs.
Imperfect-information settings require one to mix their strategies in order to be optimal, which introduces new challenges not addressed in this paper.
One would need to bound the increase in exploitability due to the watermark.

\section*{Acknowledgements}

This work has been supported by the Vannevar Bush Faculty Fellowship ONR N00014-23-1-2876, National Science Foundation grant RI-2312342, and NIH award A240108S001.
Any opinions, findings, and conclusions or recommendations expressed in this material are those of the authors and do not necessarily reflect the views of the funding agencies.

\bibliographystyle{abbrvnat}
\bibliography{neurips_2026}

%%%%%%%%%%%%%%%%%%%%%%%%%%%%%%%%%%%%%%%%%%%%%%%%%%%%%%%%%%%%

\newpage
\appendix

\section{Pseudocode of the KGW watermark}
\label{sec:kgw-watermark}

The pseudocode of the KGW watermark is shown in Algorithm~\ref{alg:kgw-watermark}.
A textual explanation of the pseudocode is available in Section~\ref{subsec:watermarking-llms}.

\begin{algorithm}[ht!]
	\caption{KGW watermark for text generation}
	\label{alg:kgw-watermark}
	\textbf{Input:} LLM with vocabulary $\mathcal{V}$, \textcolor{Green}{\textbf{green}} list size $\gamma \in (0, 1)$, and hardness $\delta \ge 0$.
	\BlankLine
	\textsc{NextToken}(prompt $v_1, v_2, \hdots, v_{t - 1} \in \mathcal{V}$): \\
	\Indp{
		Query the LLM with prompt $v_1, v_2, \hdots, v_{t - 1}$ to obtain logits $\vec{l} \in \mathbb{R}^\mathcal{V}$ over the vocabulary. \\
		\BlankLine
		Seed a pseudo-random number generator with the hash of $v_{t - 1}$. \\
		\BlankLine
		Randomly partition $\mathcal{V}$ into a \textcolor{Green}{\textbf{green}} list $G$ of size $\gamma|\mathcal{V}|$ and a \textcolor{Red}{\textbf{red}} list $R$ of size $(1 - \gamma)|\mathcal{V}|$. \\
		\BlankLine
		$\vec{p} \coloneq \softmax\left( \vec{l} + \left(\begin{cases} \delta & v \in G \\ 0 & v \in R \end{cases}\right)_{v \in \mathcal{V}} \right)$\;
		\BlankLine
		Sample the next token $v_t \in \mathcal{V}$ from $\vec{p}$. \\
		\BlankLine
		\Return $v_t$\;
	}
	\Indm
\end{algorithm}

\section{Omitted proofs}

This section contains the proofs omitted from the main body of the paper.

\subsection{Proof of Theorem~\ref{thm:concentration-bound}}
\label{sec:concentration-bound}

\begin{proof}
	We can formalize this setting as follows: let $X_1, \ldots, X_n \in \{0, 1\}$ be independent random variables where, for each $i \in [n]$, $X_i = 1$ if the agent chooses a \textcolor{Green}{\textbf{green}}-list action as the $i$\textsuperscript{th} action or $X_i = 0$ if otherwise.
	As per our assumption, we have that $\forall i \in [n]: \Pr[X_i = 1] = p$.
	Denote the number of times the agent selects a \textcolor{Green}{\textbf{green}}-list action by $X = \sum_{i = 1}^n X_i$.
	Note that
	\[
		\mathbb{E}[X] = \mathbb{E}\left[\sum_{i = 1}^n X_i\right] = \sum_{i = 1}^n \Pr[X_i = 1] = \sum_{i = 1}^n p = pn.
	\]
	By Hoeffding's inequality~\cite{hoeffding}, we have that, for any $t > 0$,
	\[
		\Pr[X \ge pn + t] < e^{-2t^2/n^2}.
	\]
	The loss in utility the agent incurs by choosing a single \textcolor{Green}{\textbf{green}}-list action is at most $\delta + \gamma\delta/(1 - \gamma) = \delta (1 + \gamma/(1 - \gamma))$.
	By definition, we have that $L \le X \delta (1 + \gamma/(1 - \gamma))$, so $X < pn + t \implies L < \delta (1 + \gamma/(1 - \gamma)) (pn + t)$.
	Then,
	\[
		\Pr[L < \delta (1 + \gamma/(1 - \gamma)) (pn + t)] \ge \Pr[X < pn + t] = 1 - \Pr[X \ge pn + t] > 1 - e^{-2t^2/n^2},
	\]
	as required.
\end{proof}

\subsection{Proof of Proposition~\ref{pro:bound}}
\label{sec:bound}

\begin{proof}
	Let $n_G$ be the number of \textcolor{Green}{\textbf{green}}-list actions played by the agent.
	The loss in utility the agent incurs by choosing a single \textcolor{Green}{\textbf{green}}-list action is at most $\delta + \gamma\delta/(1 - \gamma) = \delta (1 + \gamma/(1 - \gamma))$.
	Since $n_G \le n$, we have that
	\[
		L \le n_G \delta (1 + \gamma/(1 - \gamma)) \le n \delta (1 + \gamma/(1 - \gamma)),
	\]
	as required.
\end{proof}

\section{Testbench specification}
\label{sec:resources}

Our testbench contains an Intel Core i5-4690, a 4-thread desktop processor, and 16 GB of memory.

%%%%%%%%%%%%%%%%%%%%%%%%%%%%%%%%%%%%%%%%%%%%%%%%%%%%%%%%%%%%

% \newpage
% \input{checklist.tex}

\end{document}